\documentclass[12pt]{amsart}
\usepackage[applemac]{inputenc}
\usepackage[titletoc]{appendix}
\usepackage{bm}
\usepackage{pb-diagram}
\usepackage{blindtext}
\usepackage{amssymb}
\usepackage{amsfonts}
\usepackage{amsthm}
\usepackage{tikz}
\usepackage{tikz-cd}
\usepackage{caption}
\usepackage{color}
\usepackage{graphicx}
\usepackage[dvipsnames]{xcolor}
\usepackage{shadethm}
\usepackage{subfigure}

\usepackage{epigraph}
\usepackage{tikz}
\usetikzlibrary{tqft}
\usetikzlibrary{arrows}
\usetikzlibrary{cd}
\tikzset{>=latex}

\usepackage{hyperref}
\usepackage{amsrefs}

\newcommand{\NN}{{\mathbb{N}}}

\newcommand{\Lie}         {\mathcal L}

\newcommand{\calS}{\mathcal{S}}

\newcommand{\calM}{\mathcal{M}}
\newcommand{\calN}{\mathcal{N}}

\newcommand{\calF}{\mathcal{F}}

\newcommand{\Cour}[1]     {[\![#1]\!]}

\def\gpd{\,\lower1pt\hbox{$\longrightarrow$}\hskip-.24in\raise2pt
               \hbox{$\longrightarrow$}\,}

\newcommand{\Z}{\mathbb{Z}}

\newcommand{\E}{\epsilon}

\DeclareMathOperator{\Map}{Map}

\DeclareMathOperator{\Hom}{Hom}
\DeclareMathOperator{\ev}{ev}

\newcommand{\D}{\mathcal{D}}

\makeatletter

\pgfdeclareshape{genus}{
 \anchor{center}{\pgfpointorigin}
\backgroundpath{
    \begingroup
    \tikz@mode
    \iftikz@mode@fill

         \pgfpathmoveto{\pgfqpoint{-0.78cm}{-.17cm}}
     \pgfpathcurveto 
            {\pgfpoint{-0.35cm}{-.44cm}}
        {\pgfpoint{0.35cm}{-.44cm}}
        {\pgfpoint{.78cm}{-0.17cm}} 
     \pgfpathmoveto{\pgfqpoint{-0.78cm}{-0.17cm}}
     \pgfpathcurveto 
            {\pgfpoint{-0.25cm}{.25cm}}
        {\pgfpoint{.25cm}{.25cm}}
        {\pgfpoint{0.78cm}{-0.17cm}}
        \pgfusepath{fill}
        \fi

        \iftikz@mode@draw
     \pgfpathmoveto{\pgfqpoint{-1cm}{0cm}}
     \pgfpathcurveto 
            {\pgfpoint{-0.5cm}{-.5cm}}
        {\pgfpoint{0.5cm}{-.5cm}}
        {\pgfpoint{1cm}{0cm}}

     \pgfpathmoveto{\pgfqpoint{-0.75cm}{-0.15cm}}
     \pgfpathcurveto 
            {\pgfpoint{-0.25cm}{.25cm}}
        {\pgfpoint{.25cm}{.25cm}}
        {\pgfpoint{0.75cm}{-0.15cm}}
              \pgfusepath{stroke}
        \fi
        \endgroup
    }
    }

\newcommand{\R}{\mathbb{R}}
\newcommand{\N}{\mathbb{N}}

\DeclareMathOperator{\im}{im}
\renewcommand{\tilde}{\widetilde}

\newtheorem{thm}{Theorem}[section]
\newtheorem{prop}[thm]{Proposition}
\newtheorem{lemma}[thm]{Lemma}

\theoremstyle{definition}
\newtheorem{definition}[thm]{Definition}

\newtheorem{example}[thm]{Example}

\numberwithin{equation}{section}

\begin{document}

\title [Graded Contact Geometry and AKSZ]{Graded Contact Geometry and the AKSZ Formalism}
\author [I. Contreras]{Ivan Contreras}
\address{Department of Mathematics\\
Amherst College\\
31 Quadrangle Drive\\
Amherst, MA 01002, USA}
\email{icontreraspalacios@amherst.edu}
\author [N. Martinez Alba]{Nicolas Martinez Alba}
\address{Universidad Nacional de Colombia\\
Sede Bogot\'a-- Facultad de Ciencias\\
Departamento de Matem\'aticas\\
Carrera 45 No 26-85, Bogot\'a, 111321,  Colombia}
\email{nmartineza@unal.edu.co}

\author [R. Mehta]{Rajan Amit Mehta}
\address{Department of Mathematical Sciences\\
Smith College\\
44 College Lane\\
Northampton, MA 01063, USA}
\email{rmehta@smith.edu}


\begin{abstract}

The AKSZ formalism is a construction of topological field theories where the target spaces are differential graded symplectic manifolds. In this paper, we describe an analogue of the AKSZ formalism where the target spaces are differential graded contact manifolds. We show that the space of fields inherits a weak contact structure, and we construct a solution to the analogue of the classical master equation, defined via the Jacobi bracket. In the $n=1$ case, we recover the Jacobi sigma model, and in the $n=2$ case, we obtain three-dimensional topological field theories associated to Courant-Jacobi algebroids.

\end{abstract}

\maketitle

\section{Introduction}
Topological Field Theories have been studied extensively over the past decades. In particular, the formulation using supergeometry, known as the AKSZ formalism \cite{AKSZ} has been fundamental in the study of field theories with boundary \cite{cattaneo2018perturbative}, deformation quantization \cite{kontsevich2003deformation}, the Poisson sigma model \cite{PSM2}, among others.

The AKSZ method constructs solutions of the BV
master equation directly, without the need of starting from a classical action with a set of symmetries, in comparison to the BV formalism. The classical action
can be recovered by setting the nonzero degree fields  to
zero. In the AKSZ formalism, the target space is assumed to be a differential graded symplectic manifold. In low degrees, the construction leads to the Poisson sigma model and the Courant sigma model, a special case of which is Chern-Simons theory.

In this manuscript, we describe an analogue of the AKSZ formalism where the target space is a differential graded contact manifold. We prove results that are analogous to those in ordinary AKSZ theory, namely that the space of fields inherits a graded contact structure, and that there is a naturally-defined function that satisfies an analogue of the classical master equation, defined via Jacobi brackets.

As in ordinary AKSZ theory, the classical action can be recovered by setting the nonzero degree fields to zero. In low degrees, this construction leads to the Jacobi sigma model \cite{bascone2021jacobi} and a (to our knowledge, new) three-dimensional theory associated to a Courant-Jacobi algebroid.

We note that the existing literature on the Jacobi sigma model (e.g.\ \cites{bascone2021jacobi, bpv:jacobitwisted, cosmo2025jacobi}) either relies on Poissonization, extracting the theory from the Poisson sigma model, or directly introduces the action, justifying it by observing that it has the desired equations of motion. The AKSZ-contact theory provides a more conceptual and direct explanation of where the Jacobi sigma model comes from.

The structure of the manuscript is as follows. In Section \ref{sec:ct}, we provide an overview on contact geometry. In Section \ref{sec:graded}, we review graded contact structures. In Section \ref{sec:akszcontact}, we describe the AKSZ-contact formalism and prove our main results (Theorem \ref{thm:transgerssion2} and Theorem \ref{thm:ct-CME}). In Section \ref{sec:jacobi}, we obtain the Jacobi Sigma Model as the $n=1$ case of our construction. Finally, in Section \ref{sec:contactcourant}, we consider the $n=2$ case, leading to the Courant-Jacobi sigma model.

In work in progress \cite{JSM-reduction}, we use the AKSZ formulation of the Jacobi Sigma Model to provide a path space construction (via phase space reduction) of contact groupoids integrating Jacobi structures, including their shifted version \cite{maglio2024shifted}.

\subsection*{Acknowledgments} I.C. thanks Universidad Nacional de Colombia and N.M.A. thanks Amherst College for the hospitality. 

\section{Overview of contact geometry}\label{sec:ct}

In this section, we give a brief overview of contact geometry, covering notation and basic results that will be used later in the paper. For more details, see \cites{Geiges, mcduff_salamon}.

\subsection{Contact structures and contact forms}

Let $M$ be a finite-dimensional manifold, and let $\D \subseteq TM$ be a corank $1$ distribution on $M$. Locally, $\D$ can be expressed as the kernel of a $1$-form $\alpha$. If such a $1$-form is chosen, then the $2$-form $d\alpha$ induces a map $(d\alpha)^\flat : TM \to T^*M$, $X \mapsto i_X d\alpha$.

\begin{definition}
    The distribution $\D$ is called a \emph{contact structure} if the map
    \[ \D \hookrightarrow TM \xrightarrow{(d\alpha)^\flat} T^*M \twoheadrightarrow \D^*\]
    is an isomorphism. In this case, the $1$-form $\alpha$ is called a \emph{contact form}.
\end{definition}

One can check that the property of being a contact structure is independent of the choice of $1$-form $\alpha$. Contact forms can be characterized in several different ways.

\begin{prop}[\cite{Geiges}] \label{prop:contact}

Let $M$ be a finite-dimensional manifold, and let $\alpha$ be a (possibly only locally defined) nonvanishing $1$-form. Then the following are equivalent:
\begin{enumerate}
\item $\alpha$ is a contact form.
\item The vector bundle morphism $\flat:TM\to T^*M$ defined by 
\[\flat(X)=(i_X\alpha)\alpha+i_Xd\alpha\]
is an isomorphism.
\item The map 
\begin{align*}
    \mathfrak{X}(M) &\to \Gamma(\im((d\alpha)^\flat))\oplus C^\infty(M) \\
    X&\mapsto (i_Xd\alpha,i_X\alpha)
\end{align*} 
is an isomorphism of $C^\infty(M)$-modules.
\end{enumerate}
\end{prop}

If a contact form for a contact structure $\D$ exists globally, then $\D$ is said to be \emph{coorientable}. Throughout this paper, we will assume coorientability unless stated otherwise. If $\alpha$ is a contact form and $f$ is any nonvanishing smooth function, then $\alpha' = f\alpha$ is another contact form for the same contact structure. It should be stressed that the isomorphisms in parts (2) and (3) of Proposition \ref{prop:contact} depend on the choice of contact form.

\subsection{Contact vector fields}

Let $M$ be a manifold with contact structure $\D$. A vector field $X \in \mathfrak{X}(M)$ is called a \emph{contact vector field} if it preserves the contact structure, i.e.\ if $[X,Y] \in \Gamma(\D)$ for all $Y \in \Gamma(\D)$. If $\alpha$ is a contact form, then we can equivalently characterize contact vector fields by the property that
\[ L_X \alpha = g \alpha\] 
for some $g \in C^\infty(M)$.

An important example of a contact vector field is the \emph{Reeb vector field} $R$, which is defined by the equation $\flat(R)=\alpha$. Equivalently, $R$ is the unique solution to the equations
\begin{align*} 
i_R \alpha &= 1, & i_R d\alpha &= 0. \end{align*}
We emphasize that the Reeb vector field depends on the choice of contact form $\alpha$.

The following lemma provides a bijection between smooth functions and contact vector fields.

\begin{lemma}[Lemma~3.49 in \cite{mcduff_salamon}]\label{lemma:contactvf}
Let $\alpha$ be a contact form with Reeb vector field $R$. Then $X\in \mathfrak{X}(M)$ is a contact vector field if and only if there exists $f \in C^\infty(M)$ such that\footnote{This statement differs from the reference by a sign, but we adopt it to agree with the sign in the graded version.} 
\begin{align}\label{eq:Hamv.f}
i_X\alpha&=f, & i_Xd\alpha&=R(f)\alpha - df.
\end{align}
Conversely, for every $f \in C^\infty(M)$, the vector field $X_f$ given by \[\flat(X_f) = (f + R(f))\alpha - df\] 
is a contact vector field that satisfies \eqref{eq:Hamv.f}.
\end{lemma}

The vector field $X_f$ is called the \emph{contact Hamiltonian vector field} of $f$. From \eqref{eq:Hamv.f}, we have
\begin{equation}\label{eq:infsym-ham}
L_{X_f}\alpha= R(f)\alpha.
\end{equation}

\subsection{The Jacobi and Cartan brackets}
Let $M$ be a manifold with a contact form $\alpha$. The space of contact vector fields is closed under the Lie bracket, so the correspondence of Lemma \ref{lemma:contactvf} induces a Lie bracket $\{\cdot,\cdot\}_J$ on $C^\infty(M)$, given by
\[ X_{\{f,g\}_J} = [X_f, X_g].\]
This bracket is called the \emph{Jacobi bracket}\footnote{There are also places in the literature where it has been called the \emph{Lagrange bracket}.}. Using \eqref{eq:Hamv.f} and \eqref{eq:infsym-ham}, we can obtain the more explicit formula
\begin{equation}
\label{eq:ct-brk}
\{f,g\}_J = i_{[X_f,X_g]}\alpha = L_{X_f} i_{X_g} \alpha - i_{X_g} L_{X_f} \alpha = X_f(g) - R(f) g.
\end{equation}

The Jacobi bracket can be viewed as a contact-geometric analogue of the Poisson bracket in symplectic geometry. However, unlike the Poisson bracket, the Jacobi bracket fails to be a derivation, instead satisfying
\[ \{f,gh\}_J = \{f,g\}_J h + g\{f,g\}_J + R(f) gh.\]
Alternatively, we could consider the \emph{Cartan bracket} $\{\cdot,\cdot\}_C$, given by
\begin{equation}\label{eq:ct-brk-2}
    \{f,g\}_C=X_f(g) = \{f,g\}_J + R(f)g,
\end{equation}
which is a derivation in the second entry, but does not satisfy skew-symmetry or the Jacobi identity. On the subspace 
\[C^\infty_R(M)=\{f\in C^\infty(M):R(f)=0\},\]
of functions that are invariant under $R$, the Jacobi and Cartan brackets coincide, giving $C^\infty_R(M)$ the structure of a Poisson algebra.

\subsection{Symplectization}
\emph{Symplectization} is a process that produces a symplectic manifold from a contact manifold (with a fixed contact form). Given a manifold $M$ with contact form $\alpha$, set $\tilde{M} = M \times \R$, and let 
\[\omega = d(e^t \alpha) = e^t dt \wedge \alpha + e^t d\alpha,\]
where $t$ is the coordinate on $\R$. Then $\tilde{M}$ is a symplectic manifold equipped with a distinguished vector field $Z = \frac{\partial}{\partial t}$. The symplectic form is \emph{homogeneous} (with respect to $Z$), in the sense that it satisfies the equation $L_Z \omega = \omega$. 

On the other hand, given a homogeneous symplectic form $\omega$ on $\tilde{M}$, one can obtain a contact form on $M$ via the equation $\alpha = i_Z \omega |_{t=0}$. The following statement summarizes this correspondence.

\begin{prop}\label{prp:symplectization}
There is a one-to-one correspondence between contact forms on $M$ and homogeneous symplectic forms on $\tilde{M}$.
\end{prop}

Many structures on a contact manifold can be lifted to its symplectization. Some such lifts and their properties are described in the following statement; see \cite{mcduff1991symplectic} for more details. 

\begin{prop}\label{prop:prop-symp}
Let $M$ be a manifold with contact form $\alpha$, and let $(\tilde{M}, \omega)$ be its symplectization. 
\begin{enumerate}

\item For any $f \in C^\infty(M)$, the vector field
\[\widetilde{X_f}=X_f-R(f)\frac{\partial}{\partial t}\]
on $\tilde{M}$ is the Hamiltonian vector field associated to the function $\tilde{f} = e^t f$.
\item The map $C^\infty(M) \to C^\infty(\tilde{M})$, $f \mapsto \tilde{f} = e^t f$, satisfies 
\[\widetilde{\{f,g\}}_J = \{\tilde{f},\tilde{g}\}_\omega,\] 
where $\{\cdot, \cdot\}_\omega$ is the Poisson bracket associated to $\omega$.
\end{enumerate}
\end{prop}

\subsection{Weak contact manifolds}

We will sometimes need to consider contact structures on (e.g.\ Banach or Fr\'echet) infinite-dimensional manifolds. In these cases, we define a \emph{weak contact structure} to be a corank $1$ distribution $\D$ such that the map in Definition \ref{def:isom. ctct} is injective.  This definition parallels the notion of weak symplectic structures (see e.g.\cite{chernoff2006properties}).

One needs to be careful when working with weak contact structures, as the Reeb vector field and, more generally, contact vector fields $X_f$, may not exist. But when they do exist, they are unique.

\section{Graded contact geometry}\label{sec:graded}

This section contains a brief review of graded manifolds and graded contact structures. For more details on graded manifolds, see \cites{cattaneo-schaetz, fairon, mehta:thesis}. Graded contact structures are described in \cites{mehta2013, grabowski:gradedcontact}, and the closely related theory of graded symplectic structures is in \cite{roytenberg:symplectic}.

\subsection{Graded manifolds}
A \emph{graded manifold} $\mathcal M$ is a manifold $M$ equipped with a sheaf $\mathcal{O}_\mathcal{M}$ of graded-commutative $\Z$-graded algebras such that, locally,
\begin{equation}\label{eqn:localgraded} 
\mathcal{O}_\mathcal{M}(U) \cong C^\infty(U) \otimes S(W),
\end{equation}
where $W$ is a $\Z$-graded vector space. It is usually assumed that $W$ does not have a degree $0$ component. 

We think of $\mathcal{O}_\mathcal{M}$ as being the sheaf of ``smooth functions'' on $\mathcal{M}$, and we will use $C(\mathcal{M})$ to denote the global sections of $\mathcal{O}_\mathcal{M}$. The local form \eqref{eqn:localgraded} implies that there exist graded local coordinates, where the degree $0$ coordinates are coordinates on $M$, the functions are smooth in the degree $0$ coordinates and polynomial in the nonzero degree coordinates, and odd degree functions anti--commute with each other. We will use the notation $|f|$ to denote the degree of a homogeneous function.

We will often assume that the sheaf $\mathcal{O}_\mathcal{M}$ is nonnegatively graded. In this case, $\mathcal{M}$ is said to be an \emph{$\N$-manifold}.

The following basic example will often appear in this paper.
\begin{example}\label{ex:OddTangent} \label{ex:oddCT} (Shifted vector bundles). If $E \to M$ is a vector bundle, then the $\N$-manifold $E[1]$ is given by $\mathcal{O}_{E[1]} = \wedge \Gamma(E^*)$. In particular, the algebra of functions on $T[1]M$ is the algebra $\Omega(M)$ of differential forms on $M$, and the algebra of functions on $T^*[1]M$ is the algebra $\mathfrak{X}^\bullet(M) = \wedge \Gamma(TM)$ of multivector fields on $M$.
\end{example}

Many differential geometric constructions on manifolds can be defined more generally on graded manifolds. In particular, a (degree $k$) vector field on $\mathcal{M}$ is defined as a (degree $k$) graded derivation of $\mathcal{O}_\mathcal{M}$. The space $\mathfrak{X}(\mathcal{M})$ of vector fields on $\mathcal{M}$, endowed with the graded commutator bracket
\[ [X,Y] = XY - (-1)^{|X||Y|} YX,\]
is a graded Lie algebra.

\begin{example}[Euler vector field]
On any graded manifold, there is a natural degree $0$ vector field called the \emph{Euler vector field} $\E$, given by $\E(f) = |f| f$ for any homogeneous function $f$.
\end{example}

\begin{definition}
A \emph{homological vector field} is a degree $1$ vector field $Q$ such that $[Q,Q]=2Q^2 = 0$. A graded manifold ($\N$-manifold) $\mathcal{M}$ equipped with a homological vector field is called a \emph{$Q$-manifold} ($\N Q$-manifold)\footnote{The term \emph{differential graded manifold}, or \emph{dg-manifold} for short, is also often used for a graded manifold equipped with a homological vector field.}.
\end{definition}

\begin{example}\label{ex:tone}
Let $M$ be a manifold. In Example \ref{ex:OddTangent}), we saw that the algebra of functions on $T[1]M$ is $\Omega(M)$. The de Rham operator $d$ is a degree $1$ derivation satisfying $d^2=0$, so it can be viewed as a homological vector field $Q_d$ on $T[1]M$. The use of the notation $Q_d$ here is intended to avoid confusion with the de Rham operator on $\Omega(T[1]M)$.
\end{example}
\begin{example}
As a generalization of Example \ref{ex:tone}, suppose that $A \to M$ is a Lie algebroid. The Lie algebroid differential $d_A$ is a degree $1$ derivation of the algebra $\wedge \Gamma(A^*)$ of ``Lie algebroid forms'' and satisfies $d_A^2 = 0$, so it can be viewed as a homological vector field on $A[1]$. This observation is due to Va\u{\i}ntrob \cite{vaintrob}.
\end{example}

Differential forms can also be defined on a graded manifold. We refer the reader to \cite{cattaneo-schaetz} for a brief summary and to \cite{mehta:thesis} for a more extended treatment. The algebra of differential forms is bigraded, in the sense that it decomposes into a sum of spaces $\Omega^p_\ell(\mathcal{M})$ of $p$-forms of degree $\ell$. The commutation rules arise from the parity of the total degree $p+\ell$.

On $\Omega(\mathcal{M})$, one can define the Cartan calculus, which consists of the de Rham differential $d$ and, for each vector field $X \in \mathfrak{X}(\mathcal{M})$, a contraction operator $i_X$ and a Lie derivative operator $L_X$, satisfying the graded commutation rules
\begin{align*}
    [d,d] &= 2d^2 = 0, \\
    [L_X, d] &= L_X d - (-1)^{|X|} d L_X = 0, \\
    [i_X, d] &= i_X d + (-1)^{|X|} di_X = L_X, \\
    [L_X, L_Y] &= L_X L_Y - (-1)^{|X||Y|} L_Y L_X = L_{[X,Y]},\\
    [L_X, i_Y] &= L_X i_Y - (-1)^{|X|(|Y|-1)} i_Y L_X = i_{[X,Y]},\\
    [i_X,i_Y] &= i_X i_Y - (-1)^{(|X|-1)(|Y|-1)} i_Y i_X = 0.
\end{align*}

\subsection{Contact structures on graded manifolds}\label{sec:ctctgraded}

Much of the theory of contact manifolds extends to the graded setting. We give a brief summary here and refer the reader to \cite{mehta2013} for more details. We also mention \cite{grabowski:gradedcontact}, where the theory is developed completely in terms of the distribution, allowing for non-coorientable contact structures.

Let $\mathcal{M}$ be an $\NN$-manifold, let $\alpha \in \Omega^1_n(\mathcal{M})$ be a nowhere-vanishing $1$-form of degree $n$, and let $\D = \ker \alpha \subset T\mathcal{M}$.

\begin{definition}
The $1$-form $\alpha$ is a \emph{contact form} if the map
 \begin{equation}\label{def:isom. ctct}
    \D \hookrightarrow T\calM \xrightarrow{(d\alpha)^\flat} T^*\calM \twoheadrightarrow \D^*
 \end{equation}
 is an isomorphism.
\end{definition}

In principle, one should define a contact $\N$-manifold in terms of the distribution $\D$, as in \cite{grabowski:gradedcontact}, but in this paper we will use the term (degree $n$) \emph{contact $\N$-manifold} to mean an $\N$-manifold equipped with a contact form (of degree $n$).

The statement of Proposition \ref{prop:contact} extends directly to the graded setting (see Lemma~2.1 in \cite{mehta2013}), with the minor caveat that the isomorphism $\flat: T\mathcal{M} \to T^*\mathcal{M}$, $X \mapsto (i_X \alpha)\alpha + i_X d\alpha$, is nonhomogeneous.

As in the ungraded case, the Reeb vector field $R$ is uniquely defined by $\flat(R) = \alpha$, or equivalently, by $i_R \alpha = 1$, $i_R d\alpha =0$. We note that $|R| = -n$.

The statement of Lemma \ref{lemma:contactvf} also extends to the graded setting; see Proposition~2.4 in \cite{mehta2013}. Using the sign conventions there, the contact vector field $X_f$ associated to a homogeneous function $f \in C^\infty(\mathcal{M})$ is given by
\begin{align}\label{eqn:grad}
    i_{X_f} \alpha &= f, & i_{X_f}d\alpha = (-1)^{n(|f|-1)}R(f)\alpha - (-1)^{|f|-n} df.
\end{align}
We note that $|X_f| = |f|-n$. From \eqref{eqn:grad}, it follows that
\[ L_{X_f} \alpha = (-1)^{n(|f|-1)} R(f)\alpha.\]
The formulas for the Jacobi and Cartan brackets are then
\begin{align}
\{f,g\}_J &= i_{[X_f,X_g]} \alpha = X_f(g) - (-1)^{n(|f|+1)} R(f) g,\\
\{f,g\}_C &= X_f(g).
\end{align}

\begin{definition}
A \emph{contact $\N Q$-manifold} is a contact $\N$-manifold equipped with a vector field $Q$ that is contact and homological.
\end{definition}

Using the correspondence between functions and contact vector fields, we can equivalently say that a degree $n$ contact $\N Q$-manifold is a degree $n$ contact $\N$-manifold equipped with a degree $n+1$ function $S$ such that $\{S,S\}_J = 0$.

\subsection{The structure of contact $\N$-manifolds}
\label{sec:structure}

The following result from \cite[Theorem~2.9]{mehta2013} provides a correspondence between contact structures and symplectic structures on $\N$-manifolds.

\begin{thm}\label{thm:ctctSchw}
Suppose that $\calM$ is an $\N$-manifold equipped with a degree $n$ contact form with $n>0$. Then there is a canonical isomorphism
\[ \calM \cong \calN \times \R[n],\]
where $\calN$ is a symplectic $\N$-manifold.
\end{thm} 

There is an explicit description of how the contact form $\alpha$ on $\calM$ is related to the symplectic form $\omega$ on $\calN$, given by
\[ \alpha = \frac{1}{n}(i_\epsilon \omega + d\theta),\]
where $\epsilon$ is the Euler vector field and $\theta = i_\epsilon \alpha$ is the coordinate on $\R[n]$, called the \emph{Euler function}.

It should be stressed that the correspondence of Theorem \ref{thm:ctctSchw} relates a contact $\N$-manifold to a \emph{smaller} symplectic $\N$-manifold. This is in contrast to symplectization (see Section \ref{sec:symplectization}), which relates a contact $\N$-manifold to a larger symplectic $\N$-manifold.

The result of Theorem \ref{thm:ctctSchw} allows us to make use of results from graded symplectic geometry in the study of graded contact geometry. In particular, there is the Darboux Theorem \cite{schwarz:bv} and numerous useful facts from \cite{roytenberg:symplectic}, including the classification of symplectic $\N$-manifolds of degrees $1$ and $2$, which will play an important role in Sections \ref{sec:jacobi} and \ref{sec:contactcourant}.
\subsection{Symplectization}\label{sec:symplectization}
Here we describe the graded version of symplectization and the corresponding generalizations of the results in Proposition~\ref{prop:prop-symp}. 

Let $\calM$ be a contact $\N$-manifold with degree $n$ contact form $\alpha$. Set $\tilde{\calM} = \calM \times \R$, and let
\[ \omega = d(e^t \alpha) = e^t(dt \wedge \alpha + d\alpha),\]
where $t$ is the coordinate on $\R$. Then $\tilde{\calM}$ is a degree $n$ symplectic manifold equipped with a distinguished vector field $Z = \frac{\partial}{\partial t}$, satisfying the homogeneity condition $L_Z \omega = \omega$.

The following is the generalization of Proposition \ref{prop:prop-symp} to the graded setting.
\begin{prop}\label{prop:prop-sympgraded}
\begin{enumerate}
\item For any $f \in C^\infty(\calM)$, the vector field
\[\widetilde{X_f}=X_f-(-1)^{n(|f|-1)}R(f)\frac{\partial}{\partial t}\]
on $\tilde{\calM}$ is the Hamiltonian vector field associated to the function $\tilde{f} = e^t f$.
\item The map $C^\infty(\calM) \to C^\infty(\tilde{\calM})$, $f \mapsto \tilde{f} = e^t f$, satisfies 
\[\widetilde{\{f,g\}}_J = \{\tilde{f},\tilde{g}\}_\omega,\] 
where $\{\cdot, \cdot\}_\omega$ is the Poisson bracket associated to $\omega$.
\end{enumerate}
\end{prop}
\begin{proof}
By direct calculation, we have
\[ d\tilde{f} = e^t(dt \cdot f + df) = (-1)^{|f|-n-1}i_{\tilde{X_f}}\omega,\]
which proves the first statement. For the second statement, we have
\begin{equation}
    \begin{split}
       \{\tilde{f},\tilde{g}\}_\omega &= X_{\tilde{f}}(\tilde{g}) \\
       &= e^t\left(X_f(g) - (-1)^{n(|f|-1)}R(f)g\right) \\
       &= e^t \{f,g\}_J\\
       &= \widetilde{\{f,g\}}_J.\qedhere
    \end{split}
\end{equation}
\end{proof}

A consequence of Proposition \ref{prop:prop-sympgraded} is that, if $\calM$ is a degree $n$ contact $\N Q$-manifold with homological contact vector field $Q = X_S$, then the symplectization $\tilde{\calM}$ is a degree $n$ symplectic $\N Q$-manifold with Hamiltonian function $\tilde{S} = e^t S$. 

Via a graded version of Proposition \ref{prp:symplectization}, one can characterize contact $\N$-manifolds and contact $\N Q$-manifolds in terms of their symplectizations. This is the approach taken in \cite[Theorem 5.1]{grabowski:gradedcontact}.

\section{The AKSZ-contact formalism}\label{sec:akszcontact}

This section contains the main results of the paper, where we describe a contact-geometric analogue of the AKSZ formalism, which produces an $(n+1)$-dimensional topological field theory associated to any degree $n$ contact $\N$-manifold. We begin with some general background information in Sections \ref{sec:mapping} and \ref{sec:lifting}, and the AKSZ-contact formalism is described in Sections \ref{sec:transgression} and \ref{Poly-CME}. In Section \ref{sec:symplectizationaksz}, we describe a relationship between AKSZ-contact theory and the usual AKSZ theory, via symplectization. In Section \ref{sec:classical} we describe the associated classical theory.

\subsection{Mapping spaces}
\label{sec:mapping}
Let $\calM$ and $\calN$ be graded manifolds. We denote by $\Hom(\calN,\calM)$ the space of maps of graded manifolds from $\calN$ to $\calM$. This space can be endowed with the structure of a (usually infinite-dimensional Fr\'echet) manifold.

The \emph{mapping space} $\Map(\calN,\calM)$ is a graded manifold whose body is $\Hom(\calN,\calM)$, given by the property\footnote{More precisely, the functor $\Map(\calN,-)$ is defined as the right adjoint to the functor $- \times \calN$.} that
\begin{equation}\label{eqn:mapping}
\Hom(\mathcal{Z} \times \calN, \calM) = \Hom(\mathcal{Z},\Map(\calN,\calM))
\end{equation}
for all graded manifolds $\mathcal{Z}$. For more details on mapping spaces, see \cite{mnev2017lectures, roytenberg:aksz}.

A key feature of the mapping space is the existence of a canonical \emph{evaluation map}
\[ \ev: \Map(\calN, \calM) \times \calN \to \calM,\]
defined as the element of $\Hom(\Map(\calN, \calM) \times \calN, \calM)$ that corresponds to the identity map in $\Hom(\Map(\calN, \calM),\Map(\calN, \calM))$ under \eqref{eqn:mapping} with $\mathcal{Z} = \Map(\calN,\calM)$.

\subsection{Lifting of geometric structures}
\label{sec:lifting}
In this section, we summarize different ways that geometric structures can be lifted from $\calM$ or $\calN$ to the mapping space $\Map(\calN,\calM)$. For more details, we refer to \cite{cattaneo-schaetz,mnev2017lectures,roytenberg:aksz}.

The first observation is that the diffeomorphism spaces $\Map(\calN,\calN)$ and $\Map(\calM,\calM)$ are graded Lie groups that act on $\Map(\calN,\calM)$ by pre- and post-composition. By differentiating, one obtains maps of vector fields
\begin{align*}
    \mathfrak{X}(\calN) &\to \mathfrak{X}(\Map(\calN,\calM)),& \mathfrak{X}(\calM) &\to \mathfrak{X}(\Map(\calN,\calM))\\
    X &\mapsto \hat{X}, & Y &\mapsto \breve{Y},
\end{align*}
which are morphisms (or anti-morphisms, depending on sign conventions) of graded Lie algebras. Additionally, because the two actions commute, we have that $[\hat{X},\breve{Y}]=0$ for all $X \in \mathfrak{X}(\calN)$ and $Y \in \mathfrak{X}(\calM)$.

Next, we consider the lifting of differential forms from $\calM$ to the mapping space. Given any differential form $\beta \in \Omega(\calM)$, we can use the evaluation map to obtain $\ev^*\beta \in \Omega(\Map(\calN, \calM) \times \calN)$. In order to project the form down to the mapping space, we assume that $\calN$ is equipped with a Berezinian measure $\mu$. Integration with respect to the measure defines a map $\mu_*: \Omega(\Map(\calN, \calM) \times \calN) \to \Omega(\Map(\calN, \calM))$. We then define the \emph{transgression map}
\[ \mathbb{T} = \mu_* \circ \ev^*: \Omega(\calM) \to  \Omega(\Map(\calN, \calM)).\]

We will be exclusively interested in the case where $\calN = T[1]\Sigma$, where $\Sigma$ is a compact oriented $k$-dimensional manifold without boundary. In this case, $\calN$ has a canonical measure $\mu$, given by
\[ \int_{\calN} f d\mu = \int_\Sigma f\]
for $f \in C^\infty(\calN) = \Omega(\Sigma)$. By definition, this integral vanishes when $|f| < k$.

An important point is that $\mu_*$ lowers the degree by $k$ but preserves the form grading, whereas $\ev^*$ preserves both gradings. Thus the transgression map takes $\Omega^p_\ell(\calM)$ to $\Omega^p_{\ell-k}(\Map(\calN, \calM))$.

Another fact that we will use is that the lifting of vector fields from $\calM$ and the de Rham operator are compatible with the transgression map, in the sense that, for any $Y \in \mathfrak{X}(\calM)$ and $\beta \in \Omega(\calM)$,
\begin{align}\label{eqn:liftcartancalc}
    \mathbb{T}(i_Y\beta) &= i_{\breve{Y}}(\mathbb{T}(\beta)), & \mathbb{T}(d\beta) &= d(\mathbb{T}(\beta)),
\end{align}
and therefore $\mathbb{T}(L_Y\beta) = L_{\breve{Y}}(\mathbb{T}(\beta))$ as well. 

The interaction of vector fields on $\calN$ with the transgression map is more subtle and interesting. We will later use the following fact from \cite[Lemma 2.6]{cattaneo2001aksz}.
\begin{lemma}\label{lemma:invariance}
    If $X$ is a $\mu$-invariant vector field on $\calN$, then $L_{\hat{X}} (\mathbb{T}(\beta)) = 0$ for any $\beta \in \Omega(\calM)$.
\end{lemma}

\subsection{Transgression of graded contact structures}
\label{sec:transgression}
An important feature of AKSZ theory is that the space of fields inherits a geometric structure from the target space. In the usual AKSZ theory (\cite{AKSZ}, also see \cite{cattaneo2001aksz, cattaneo-schaetz, mnev2017lectures, roytenberg:aksz}), the space of fields inherits a symplectic structure from the target space. An analogous result for poly-symplectic structures is in \cite{contreras2-martinez2021poly}. Using similar techniques, we consider here the lifting of contact structures.

Suppose that $\calM$ is a degree $n$ contact $\N$-manifold with contact form $\alpha \in \Omega^1_n(\calM)$, and let $\calN = T[1]\Sigma$, where $\Sigma$ is a compact oriented $k$-dimensional manifold (without boundary). Recall that the de Rham operator on $\Omega(\Sigma)$ can be viewed as a homological vector field $Q_d$ on $\calN$.

\begin{thm}\label{thm:transgerssion2}
The form $\breve{\alpha}= \mathbb{T}(\alpha)$ is a degree $n-k$ weak contact form on $\calF= \Map(\mathcal{N},\mathcal{M})$. In addition, if $\calM$ is a degree $n$ contact $\N Q$-manifold with homological vector field $Q_\calM$, then $\calF$ is a degree $n-k$ contact $Q$-manifold with homological vector field $\hat{Q}_d \pm \breve{Q}_\calM$.

\end{thm}

\begin{proof}
The proof is similar to that of the analogous result for symplectic structures in the usual AKSZ theory. Using Theorem \ref{thm:ctctSchw} and Darboux coordinates $\{q^i,p_i\}$ on the symplectic component, we have
\begin{equation}\label{eq:alphabar}
    \breve{\alpha}=\mathbb{T}\alpha=\int_{\mathcal{N}}\left(\frac{\delta \theta}{n}+q^i \delta p_i\right) d\mu. 
\end{equation}
To show that $\breve{\alpha}$ is weak contact, we need to check that $\ker(d\breve{\alpha})$ has trivial intersection with $\ker \breve{\alpha}$. If we denote $X=X_\theta\delta \theta + X_{q^i}\delta {q^i} + X_{p_i} \delta {p_i}$, then $X$ being in $\ker(d\breve{\alpha})$ implies that $X_{q^i}=X_{p_i}=0$ for all $i$. On the other hand, $X$ being in $\ker(\breve{\alpha})$ implies that $X_\theta-q^i X_{p_i}=0$. Combining these conditions, we see that $X \in \ker(d\breve{\alpha}) \cap \breve{\alpha}$ if and only if $X=0$.

The fact that $Q = \hat{Q}_d \pm \breve{Q}_\calM$ is homological follows immediately from $Q_d$ and $Q_\calM$ being homological, together with the fact that $[\hat{Q}_d, \breve{Q}_\calM] =0$ (see Section \ref{sec:lifting}).

Since $Q_\calM$ is a contact vector field on $\calM$, it follows from \eqref{eqn:liftcartancalc} that $\breve{Q}_\calM$ is a contact vector field on $\calF$. Additionally, it follows from Lemma \ref{lemma:invariance} that $\hat{Q}_d$ is contact. We note that the $\mu$-invariance of $Q_d$ is equivalent to the more elementary fact that $\int_\Sigma d\eta = 0$ for any $\eta \in \Omega(\Sigma)$, due to Stokes' Theorem.
\end{proof}

The following is an immediate consequence of the fact (see Section \ref{sec:lifting} that Cartan calculus is preserved by the lifting of forms and vector fields from $\calM$ to $\calF$.
\begin{prop}
    Let $R$ be the Reeb vector field on $\calM$. Then $\breve{R}$ is the Reeb vector field on $\calF$.
\end{prop}

\subsection{The Classical Master Equation}\label{Poly-CME}

Suppose that $\calM$ is a degree $n$ contact $\N Q$-manifold with contact form $\alpha$ and homological vector field $Q_\calM$, and let $\calN = T[1]\Sigma$, where $\Sigma$ is a compact oriented $k$-dimensional manifold without boundary. By Theorem \ref{Poly-CME}, we obtain a degree $n-k$ contact form $\breve{\alpha}$ on $\calF = \Map(\calN,\calM)$, so we can define Jacobi and Cartan brackets
 \begin{align}\label{eq:brct on F}
     (F,G)_J&= i_{[X_F,X_G]}\breve{\alpha}, \\ \label{eq:Cbrct on F}
     (F,G)_C&= X_F(G) = (F,G)_J + (-1)^{(n-k)(|F|+1)}\breve{R}(F)G,
 \end{align}
 defined on functions $F,G \in C^\infty(\calF)$ for which the contact vector fields $X_F$ and $X_G$ exist.

Let $S_\calM= i_{Q_\calM} \alpha$ be the degree $n+1$ function on $\calM$ corresponding to the contact vector field $Q_\calM$. The transgression of $S_\calM$ is $\breve{S}_\calM = \mathbb{T}(S_\calM) = i_{\breve{Q}_\calM} \breve{\alpha}$, which is the degree $n-k+1$ function on $\calF$ corresponding to the contact vector field $\breve{Q}_\calM$. Additionally, we define $\hat{S}_\calM = i_{\hat{Q}_d} \breve{\alpha}$, which the degree $n-k+1$ function on $\calF$ corresponding to the contact vector field $\hat{Q}_d$. 

We set $S_\calF = \hat{S}_\calM + (-1)^{n+1} \breve{S}_\calM$. It is the function on $\calF$ corresponding to the contact vector field $Q_\calF = \hat{Q}_d + (-1)^{n+1} \breve{Q}_\calM$. As explained in \cite{roytenberg:aksz} in the case of ordinary AKSZ theory, the sign is chosen so the equations of motion in the induced classical theory have dg maps as solutions.

\begin{thm}\label{thm:ct-CME}
The following statements hold:
\begin{enumerate}
\item[(i)]$S_\calF$ is a solution of classical master equation 
$(S_\calF, S_\calF)_J=0$.
\item[(ii)] $S_\calF$ is a solution of the classical master equation for the Cartan bracket, $(S_\calF,S_\calF)_C=0$, if and only if $S_\calM$ is $R$-invariant. 
\end{enumerate}
\end{thm}

\begin{proof}
From \eqref{eq:brct on F}, we have
\[ (S_\calF,S_\calF)_J = i_{[Q_\calF,Q_\calF]} \breve{\alpha},\]
which vanishes because $Q_\calF$ is homological.

From \eqref{eq:Cbrct on F}, we then have
\[ (S_\calF,S_\calF)_C = (-1)^{(n-k)}\breve{R}(S_\calF)S_\calF,\]
which vanishes if and only if $\breve{R}(S_\calF)=0$. However, since $[\breve{R},\hat{Q}_d]=0$ and $L_{\breve{R}} \breve{\alpha} = \mathbb{T}(L_R \alpha) = 0$, we have that
\[ \breve{R}(\hat{S}_\calM) = L_{\breve{R}} i_{\hat{Q}_d} \breve{\alpha} = i_{[\breve{R},\hat{Q}_d]} \breve{\alpha} + i_{\hat{Q}_d}L_{\breve{R}} \breve{\alpha} = 0,\]
so $\breve{R}(S_\calF) = (-1)^{n+1}\breve{R}(\breve{S}_\calM) = (-1)^{n+1}\mathbb{T}(R(S_\calM))$, and the second statement follows.
\end{proof}

\subsection{Symplectization and AKSZ}\label{sec:symplectizationaksz}

Suppose that $\calM$ is a degree $n$ contact $\N Q$-manifold with contact form $\alpha$ and contact Hamiltonian $S_\calM$, and suppose that $\calN = T[1]\Sigma$, where $\Sigma$ is a compact oriented $k$-dimensional manifold. Then, as we saw in Section \ref{Poly-CME}, we have the function $S_\calF = \hat{S}_\calM + (-1)^{n+1} \breve{S}_\calM$ on $\calF = \Map(\calN,\calM)$, which satisfies the classical master equation $(S_\calF, S_\calF)_J=0$.

On the other hand, we can consider the symplectization $\tilde{\calM} = \calM \times \R$, which is a degree $n$ symplectic $\N Q$-manifold with symplectic form $\omega = d(e^t \alpha)$ and Hamiltonian function $S_{\tilde{\calM}} = e^t S_\calM$. Applying the usual AKSZ theory, we obtain the space of fields $\tilde{\calF} = \Map(\calN,\tilde{\calM})$, with the function $S_{\tilde{\calF}} = \hat{S}_{\tilde{\calM}} + (-1)^{n+1} \breve{S}_{\tilde{\calM}}$, where $\breve{S}_{\tilde{\calM}} = \mathbb{T}(S_{\tilde{\calM}})$ and $\hat{S}_{\tilde{\calM}} = i_{\hat{Q}_d}\breve{\lambda}$, with $\lambda = \frac{1}{n}i_\epsilon \omega$.

Consider the zero section $\tau: \calM \to \tilde{\calM}$. Applying the $\Map(\calN,-)$ functor, we obtain $\tau_*: \calF \to \tilde{\calF}$, which makes the following diagram commute:
\[ 
\begin{tikzcd}
\Omega(\tilde{\calM}) \ar[r,"\tau^*"] \ar[d,"\mathbb{T}"] & \Omega(\calM) \ar[d,"\mathbb{T}"] \\
\Omega(\tilde{\calF}) \ar[r,"(\tau_*)^*"] & \Omega(\calF)
\end{tikzcd}
\]

\begin{prop}\label{prop:symplectize}
    $(\tau_*)^*S_{\tilde{\calF}} = S_\calF$.
\end{prop}
\begin{proof}
We first observe that $\tau^* S_{\tilde{\calM}} = \tau^*(e^t S_\calM) = S_\calM$. It follows that
\begin{equation}\label{eqn:taubreve}
(\tau_*)^*\breve{S}_{\tilde{\calM}} = (\tau_*)^*\mathbb{T}(S_{\tilde{\calM}}) = \mathbb{T}(\tau^* S_{\tilde{\calM}}) = \mathbb{T}(S_\calM) = \breve{S}_\calM.    
\end{equation} 

Next, we calculate
\[ \tau^* \lambda = \frac{1}{n} i_\epsilon d\alpha = \frac{1}{n}(L_\epsilon \alpha - di_\epsilon \alpha) = \alpha - \frac{1}{n}d(i_\epsilon \alpha).\]
Recall from Section \ref{sec:structure} that the Euler function is defined as $\theta = i_\epsilon \alpha$. We can then write $\tau^* \lambda = \alpha - \frac{1}{n}d\theta$, and it follows that
\begin{equation}
\label{eqn:tauhat}
\begin{split}
    (\tau_*)^*\hat{S}_{\tilde{\calM}} &= i_{\hat{Q}_d}(\tau_*)^*\mathbb{T}(\lambda) \\
    &=  i_{\hat{Q}_d}\mathbb{T}(\tau^*\lambda) \\
    &=  i_{\hat{Q}_d}\mathbb{T}\left(\alpha - \frac{1}{n}d\theta\right) \\
    &= \hat{S}_\calM - \frac{1}{n}L_{\hat{Q}_d}\breve{\theta}\\
    &= \hat{S}_\calM,
\end{split}
\end{equation}
where the last step uses Lemma \ref{lemma:invariance}. The result follows from \eqref{eqn:taubreve} and \eqref{eqn:tauhat}.
\end{proof}

\subsection{Classical theory}\label{sec:classical}

In certain dimensions, classical elements of AKSZ theory can be extracted by restricting degree $0$ data to the base of $\calF$, which is the classical space of fields $\Hom(\calN,\calM)$. In the physics literature, this process is often described as setting the antifields to zero.

In particular, if $\Sigma$ is $(n+1)$-dimensional, then the function $S_\calF$ is of degree $0$. This induces a function, which we denote as $\calS$, on $\Hom(\calN,\calM)$. The function $\calS$ is the action functional for the classical field theory, and it can be described more explicitly as follows.

Given a map of graded manifolds $\phi: T[1]\Sigma \to \calM$, we can pull the contact form $\alpha$ back by $\phi$ to get a degree $n$ $1$-form $\phi^*\alpha$ on $T[1]\Sigma$. We can then contract with the degree $1$ vector field $Q_d$ to get a degree $n+1$ function $i_{Q_d} \phi^*\alpha$. Using the identification $C^\infty(T[1]\Sigma) = \Omega(\Sigma)$, we can then integrate over $\Sigma$ to get the \emph{kinetic part} of the action
\[ \calS_{\mathrm{kin}}(\phi) = \int_{\Sigma} i_{Q_d} \phi^* \alpha,\]
which is the restriction of $\hat{S}_\calM$ to $\Hom(\calN,\calM)$.
The restriction of $\breve{S}_\calM$ is obtained by pulling the degree $n+1$ function $S_\calM$ back by $\phi$ and then integrating. This gives the \emph{interaction part} of the action
\[ \calS_{\mathrm{int}}(\phi) = \int_{\Sigma} \phi^* S_\calM.\]
Putting the two parts together, we obtain the \emph{AKSZ-contact action functional}
\begin{equation}\label{eqn:akszaction} 
\calS(\phi) = \calS_{\mathrm{kin}}(\phi) + (-1)^{n+1} \calS_{\mathrm{int}}(\phi) = \int_\Sigma i_{Q_d} \phi^* \alpha + (-1)^{n+1} \phi^* S_\calM.
\end{equation}

The action \eqref{eqn:akszaction} is constructed directly from the contact $\N Q$-manifold structure of $\calM$. Another approach is to construct an action via the symplectization $\tilde{\calM}$, using ordinary AKSZ theory, but restricting to fields taking values in the zero section. Proposition \ref{prop:symplectize} implies that this action will be equal to \eqref{eqn:akszaction}, but it is worth describing the construction and comparing the two actions more explicitly.

The AKSZ action (see e.g.\ \cite{roytenberg:aksz}) associated to $\tilde{\calM}$ is given by
\[ \calS_{\mathrm{AKSZ}}(\psi) = \int_\Sigma i_{Q_d} \psi^* \lambda + (-1)^{n+1} \psi^* S_{\tilde{\calM}}\]
for $\psi: T[1]\Sigma \to \tilde{\calM}$, where $\lambda = \frac{1}{n} i_\epsilon \omega$ and $S_{\tilde{\calM}} = e^t S_\calM$. Given a map $\phi: T[1]\Sigma \to \calM$, we can obtain the map $\phi_0 = \tau \circ \phi: T[1]\Sigma \to \tilde{\calM}$, where $\tau: \calM \to \tilde{\calM}$ is the zero section. We then define the \emph{BPV action}\footnote{The name refers to the authors of \cite{bascone2021jacobi}, where this approach was used in the $n=1$ case to obtain the Jacobi sigma model action; see Section \ref{sec:jacobisigmamodel}.} $\calS_{\mathrm{BPV}}$ by
\[ \calS_{\mathrm{BPV}}(\phi) = \calS_{\mathrm{AKSZ}}(\phi_0).\]

To compare this action with \eqref{eqn:akszaction}, we recall from the proof of Proposition \ref{prop:symplectize} that $\tau^* S_{\tilde{\calM}} = S_\calM$ and $\tau^* \lambda = \alpha - \frac{1}{n} d\theta$.
Using these identities, we obtain the following formula for the BPV action:
\begin{equation}
    \label{eqn:bpvaction}
    \calS_{\mathrm{BPV}}(\phi) = \int_\Sigma i_{Q_d} \phi^* \alpha - \frac{1}{n}L_{Q_d} \phi^* \theta + (-1)^{n+1} \phi^* S_\calM.
\end{equation}
We note that the BPV action \eqref{eqn:bpvaction} includes an additional term that does not appear in the AKSZ-contact action \eqref{eqn:akszaction}. However, the additional term is exact and is therefore inconsequential in the integral. 

\section{The Jacobi sigma model}
\label{sec:jacobi}

In this section, we consider the $n=1$ case of AKSZ-contact theory and show that it reproduces the Jacobi sigma model of \cite{bascone2021jacobi}.

\subsection{Jacobi manifolds}
Let $M$ be a manifold. A \emph{Jacobi structure} on $M$ is a Lie bracket on $C^\infty(M)$ that is local, in the sense that the support of $\{f,g\}$ is contained in the supports of both $f$ and $g$. A theorem due to Kirillov \cite{kirillov:local} says that local Lie brackets are always first-order in each entry. From this it can be deduced that a Jacobi structure is equivalently given by a pair $(\Lambda, E)$, where $\Lambda \in \mathfrak{X}^2(M)$ is a bivector field and $E \in \mathfrak{X}(M)$ is a vector field, such that
\begin{align}\label{eqn:jacobi}
    [\Lambda, \Lambda] &= 2E \wedge \Lambda, & [\Lambda, E] &= 0,
\end{align}
where $[\cdot, \cdot]$ is the Schouten bracket of multivector fields. The bracket is then given by
\[ \{f,g\} = \Lambda(df,dg) + f E(g) - E(f) g.\]
The characterization of Jacobi structures in terms of $\Lambda$ and $E$ is due to Lichnerowicz \cite{lichnerowicz}.

Examples of Jacobi manifolds include Poisson manifolds (in the case where $E=0$), contact manifolds, cosymplectic manifolds, and locally conformal symplectic manifolds (see, e.g., \cite{LMP:cohomology}).

\subsection{Degree $1$ contact $\N Q$-manifolds}
In \cite{mehta2013}, it was shown that Jacobi manifolds are in one-to-one correspondence with degree $1$ contact $\N Q$-manifolds, in the following way. By Theorem \ref{thm:ctctSchw} and the classification of degree $1$ symplectic $\N$-manifolds in \cite{roytenberg:symplectic}, we have that every degree $1$ contact $\N$-manifold $\calM$ is canonically of the form
\[ \calM \cong T^*[1]M \times \R[1]\]
for some manifold $M$. Thus we can make the identification
\[ C^\infty(\calM) \cong \mathfrak{X}^\bullet(M)[\theta],\]
where $\theta$ is the degree $1$ coordinate on $\R[1]$. 

Any degree $2$ function $S$ on $\calM$ then has the form
\[ S = \Lambda - E \theta,\]
where $\Lambda \in \mathfrak{X}^2(M)$ and $E \in \mathfrak{X}^1(M)$. In \cite[Proposition 3.3]{mehta2013} it was shown that 
\[ \{S,S\}_J = [\Lambda,\Lambda] - 2E\Lambda + 2[E,\Lambda]\theta,\]
which vanishes if and only if the equations \eqref{eqn:jacobi} hold.

\subsection{AKSZ-contact theory in the $n=1$ case}
\label{sec:jacobisigmamodel}

We can now give an explicit description of the AKSZ-contact action in the $n=1$ case. The input data consists of a compact oriented $2$-manifold $\Sigma$ and a Jacobi manifold $M$. The fields are maps of graded manifolds
\[ \phi: T[1]\Sigma \to \calM = T^*[1]M \times \R[1].\]
In local coordinates $\{x^i, p_i, \theta\}$ on $\calM$, where $p_i$ and $\theta$ are of degree $1$, a field can be described by smooth functions $X^i = \phi^*(x^i)$ and $1$-forms $\eta_i = \phi^*(p_i)$, $\Theta = \phi^*(\theta)$ on $M$. Collectively, the functions $X^i$ form a smooth map $X: \Sigma \to M$, and the $1$-forms $\eta_i, \Theta$ form an element of $\Omega^1(\Sigma; X^*(T^*M \times \R))$.

In coordinates, the contact form on $\calM$ has the form $\alpha = p_i dx^i + d\theta$. From this, we get the kinetic part of the action:
\[ \mathcal{S}_{\mathrm{kin}}(\phi) = \int_\Sigma \eta_i  dX^i + d\Theta.\]
(Note that, for notational simplicity, we are omitting wedge product symbols for differential forms in action functionals.)
The interaction terms come from the Jacobi structure on $M$. If we write $\Lambda = \frac{1}{2}\Lambda^{ij} p_i p_j$ and $E = E^i p_i$, then we have
\[ S = \Lambda - E\theta = \frac{1}{2}\Lambda^{ij} p_i p_j - E^i p_i \theta,\]
and thus
\[ \mathcal{S}_{\mathrm{int}}(\phi) = \int_\Sigma \phi^*S = \int_\Sigma \frac{1}{2}\Lambda^{ij} \eta_i  \eta_j - E^i \eta_i  \Theta.\]
Putting the kinetic and interaction terms together, we get the AKSZ-contact action functional
\begin{equation}
\label{eqn:jacobiaction} 
\mathcal S(\phi) = 
 \int_\Sigma \eta_i  dX^i + d\Theta + \frac{1}{2}\Lambda^{ij} \eta_i  \eta_j - E^i \eta_i  \Theta.
\end{equation}
The additional term in the BPV action \eqref{eqn:bpvaction} is $-d\Theta$, so
\begin{equation}
    \label{eqn:jacobibpvaction}
  \mathcal S_{\mathrm{BPV}}(\phi) = 
 \int_\Sigma \eta_i  dX^i + \frac{1}{2}\Lambda^{ij} \eta_i \eta_j - E^i \eta_i \Theta.  
\end{equation}
This is the Jacobi sigma model action, as given in \cite{bascone2021jacobi}.

\section{The Courant-Jacobi sigma model}
\label{sec:contactcourant}

In this section, we consider the $n=2$ case of AKSZ-contact theory, which gives a $3$-dimensional topological field theory associated to Courant-Jacobi algebroids.

\subsection{Courant-Jacobi algebroids}

Courant-Jacobi algebroids generalize Courant algebroids in a similar way to how Jacobi manifolds generalize Poisson manifolds. The prototypical example (see Example \ref{example:wade} below) was given (in skew-symmetric form) by Wade \cite{wade:conformal}, and the general definition is due to Grabowski and Marmo \cite{grabowski-marmo}. The slightly more general concept of \emph{contact Courant algebroid} was introduced in \cite{grabowski:gradedcontact}. The definition we use is taken from Das \cite{das:contactcourant}, who showed that it is equivalent to that in \cite{grabowski:gradedcontact,grabowski-marmo}.

\begin{definition}\label{def:contactcourant}
A \emph{Courant-Jacobi algebroid} is a vector bundle $E \to M$ equipped with a nondegenerate symmetric bilinear form $\langle \cdot, \cdot \rangle$, a bundle map $\rho: E \to TM \times \R$, and an $\R$-bilinear bracket $\Cour{\cdot,\cdot} : \Gamma(E) \times \Gamma(E) \to \Gamma(E)$ such that
\begin{enumerate}
    \item $\Cour{e, \Cour{e',e''}} = \Cour{\Cour{e,e'},e''} + \Cour{e', \Cour{e,e''}}$,
    \item $\Cour{e,fe'} = f \Cour{e,e'} + \hat{\rho}(e)(f) e'$,
    \item $\Cour{e,e} = \frac{1}{2} \mathcal{D}\langle e,e\rangle$,
    \item $\rho(e)(\langle e',e''\rangle) = \langle \Cour{e,e'},e''\rangle + \langle e', \Cour{e,e''}\rangle$,
\end{enumerate}
for all $e,e',e'' \in \Gamma(E)$ and $f\in C^\infty(M)$, where $\mathcal{D}: C^\infty(M) \to \Gamma(E)$ is given by $\langle \mathcal{D}f,e\rangle = \rho(e)(f)$. 
Here, sections of $TM \times \R$ are viewed as acting on $C^\infty(M)$ as first-order differential operators, and $\hat{\rho}$ is the symbol of $\rho$, i.e.\ the projection onto $TM$.
\end{definition}

\begin{example}\label{example:wade}
The \emph{standard Courant-Jacobi algebroid} on a manifold $M$ is
$E = (TM \times \R) \oplus (T^*M \times \R)$, where $\rho$ is the projection onto $TM \times \R$, and the bilinear form and bracket are given by
\begin{align*}
    \langle (X_1,f_1,\xi_1,g_1), (X_2,f_2,\xi_2,g_2) \rangle &= \langle X_1, \xi_2 \rangle + \langle X_2, \xi_1 \rangle + f_1 g_2 + f_2 g_1, \\
    \Cour{(X_1,f_1,\xi_1,g_1), (X_2,f_2,\xi_2,g_2)} &= (X_3,f_3,\xi_3,g_3),
\end{align*}
where
\begin{align*}
    X_3 &= [X_1,X_2], \\
    f_3 &= X_1(f_2) - X_2(f_1),\\
    \xi_3 &= \Lie_{X_1} \xi_2 - i_{X_2} d\xi_1 + f_1 \xi_2 - f_2 \xi_1 + f_2 dg_1 + g_2 df_1,\\
    g_3 &= X_1(g_2) - X_2(g_1) + i_{X_2}\xi_1 + f_1g_2.
\end{align*}
These formulas appear in \cite[Example 12.2]{grabowski:gradedcontact}. The skew-symmetric version of this bracket appeared in \cite{wade:conformal}, where it was shown that several interesting geometric structures, including Jacobi structures, arise as special cases of Dirac structures in $E$.
\end{example}

\begin{example}
    In the case when $M$ is a point, a Courant-Jacobi algebroid consists of a vector space $V$ equipped with a nondegenerate symmetric bilinear form, a distinguished vector $z \in V$, and a bilinear bracket such that 
\begin{enumerate}
    \item $\Cour{e, \Cour{e',e''}} = \Cour{\Cour{e,e'},e''} + \Cour{e', \Cour{e,e''}}$,
    \item $\Cour{e,e} = \frac{1}{2} \langle e,e\rangle z$,
    \item $\langle z,e\rangle\langle e',e''\rangle = \langle \Cour{e,e'},e''\rangle + \langle e', \Cour{e,e''}\rangle$.
\end{enumerate}
If $z=0$, then the above conditions force the bracket to be skew-symmetric (and therefore a Lie bracket) and the bilinear form to be invariant.
\end{example}

\subsection{Degree $2$ contact $\N Q$-manifolds}

In \cite{grabowski:gradedcontact}, it was shown that Courant-Jacobi manifolds are in one-to-one correspondence with degree $2$ contact $\N Q$-manifolds. The argument there was based on symplectization. We outline an alternative approach here, using Theorem \ref{thm:ctctSchw}.

We begin by describing the correspondence in \cite{roytenberg:symplectic} between degree $2$ symplectic $\N$-manifolds and vector bundles $E \to M$ equipped with a nondegenerate bilinear form. Let $E \to M$ be such a vector bundle. Locally, if $\{x^i\}$ are coordinates on $M$ and $\{v_\alpha\}$ is a frame of sections on $E$ with dual frame $\{v^\alpha\}$, then the corresponding degree $2$ symplectic $\N$-manifold $\calN$ has coordinates $\{x^i, v^\alpha, p_i\}$, where $|v^\alpha| = 1$ and $|p_i|=2$.

Writing $g_{\alpha \beta} = \langle v_\alpha, v_\beta \rangle$, and assuming the frame is chosen so that $g_{\alpha \beta}$ is constant, the symplectic form on $\calN$ is given by
\[ \omega = dp_i \wedge dx^i + \frac{1}{2} g_{\alpha \beta} dv^\alpha \wedge dv^\beta.\]

Theorem \ref{thm:ctctSchw} tells us that every degree $2$ contact $\N$-manifold is of the form $\calM = \calN \times \R[2]$, where $\calN$ is as above. The contact form on $\calM$ is
    \begin{equation} \label{eqn:deg2alpha}
    \alpha = \frac{1}{2}\left(i_\epsilon \omega + d\theta\right) = p_i dx^i + \frac{1}{2} g_{\alpha \beta} v^\alpha d v^\beta + \frac{1}{2} d\theta,\end{equation}
    where $\theta$ is the degree $2$ coordinate on $\R[2]$.

Any degree $3$ function $S$ on $\calM$ can be written in the form
\begin{equation} \label{eqn:courantS}
S = a^i_\alpha v^\alpha p_i + b_\alpha v^\alpha \theta - \frac{1}{6} T_{\alpha \beta \gamma} v^\alpha v^\beta v^\gamma, 
\end{equation}
where $T_{\alpha \beta \gamma}$, $a^i_\alpha$, and $b_\alpha$ are functions on $M$. Such a function encodes bracket and anchor data via the equations
    \begin{align*}
        \langle \Cour{v_\alpha, v_\beta}, v_\gamma \rangle &= T_{\alpha \beta \gamma},\\
     \rho(v_\alpha)(f) &= a^i_\alpha \frac{\partial f}{\partial x^i} + b_\alpha f,
    \end{align*}
    and the equation $\{S,S\}_J = 0$ holds if and only if the axioms in Definition \ref{def:contactcourant} hold.

\subsection{AKSZ-contact theory in the $n=2$ case}

The input data for $n=2$ AKSZ-contact theory consists of a compact oriented $3$-manifold $\Sigma$ and a Courant-Jacobi algebroid $E \to M$. The fields are maps of graded manifolds 
\[\phi: T[1]\Sigma \to \mathcal{M}.\]
Locally, a field is given by smooth functions $X^i = \phi^*(x^i)$, $1$-forms $\eta^\alpha = \phi^*(v^\alpha)$, and $2$-forms $P_i = \phi^*(p_i)$ and $\Theta = \phi^*(\theta)$.

From \eqref{eqn:deg2alpha}, we get the kinetic part of the action
\[ \mathcal S_{\mathrm{kin}}(\phi) = \int_{\Sigma} P_i dX^i + \frac{1}{2} g_{\alpha \beta} \eta^\alpha d\eta^\beta + \frac{1}{2} d\Theta,\]
and from \eqref{eqn:courantS}, we get the interaction terms
\[ \mathcal S_{\mathrm{int}}(\phi) = \int_\Sigma \phi^*S = \int_{\Sigma} 
a^i_\alpha \eta^\alpha P_i + b_\alpha \eta^\alpha \Theta - \frac{1}{6} T_{\alpha \beta \gamma} \eta^\alpha \eta^\beta  \eta^\gamma.\]
 Putting the two parts together, we get the following formula for $\mathcal{S}(\phi)$:
\begin{equation}\label{eq: Courant-Jacobi}
 \int_\Sigma P_i  dX^i + \frac{1}{2} g_{\alpha \beta} \eta^\alpha d\eta^\beta + \frac{1}{2} d\Theta
    - a^i_\alpha \eta^\alpha P_i - b_\alpha \eta^\alpha \Theta + \frac{1}{6} T_{\alpha \beta \gamma} \eta^\alpha  \eta^\beta \eta^\gamma.
\end{equation}

The additional term in the BPV action is $-\frac{1}{2}d\Theta$, so the formula for $\mathcal S_{\mathrm{BPV}}(\phi)$ is
\begin{equation}\label{eq: Courant-Jacobi-BPV}
\int_\Sigma P_i dX^i + \frac{1}{2} g_{\alpha \beta} \eta^\alpha  d\eta^\beta - a^i_\alpha \eta^\alpha  P_i - b_\alpha \eta^\alpha \Theta + \frac{1}{6} T_{\alpha \beta \gamma} \eta^\alpha\eta^\beta\eta^\gamma.
\end{equation}
Under the assumption that $\Sigma$ is boundary-free, \eqref{eq: Courant-Jacobi} and \eqref{eq: Courant-Jacobi-BPV} are equal and could be called the \emph{Courant-Jacobi sigma model} action. We expect that it should provide the correct setting for explaining the appearance of Wess-Zumino-Witten terms in, e.g., the twisted Jacobi sigma model of \cite{bpv:jacobitwisted}.

\bibliography{aksz}
\end{document}